\documentclass[conference]{IEEEtran}

\usepackage[utf8]{inputenc}
\usepackage{amsmath,amssymb,bm}
\usepackage{graphicx}
\usepackage{float}
\usepackage{verbatim}
\usepackage{array}
\usepackage{cite}
\usepackage{enumerate}
\usepackage{booktabs}
\usepackage[draft]{hyperref} 
\usepackage{color}
\usepackage{float}
\usepackage[capitalise]{cleveref}
\usepackage[mathcal]{euscript}
\usepackage{multirow,multicol,tabu}
\usepackage{arydshln} 
\usepackage{subcaption}
\usepackage[left=0.625in,right=0.625in,top=0.75in,bottom=1in]{geometry}
\usepackage{soul}
\usepackage{url}
\usepackage{algorithm2e}
\SetKwComment{Comment}{/* }{ */}
\RestyleAlgo{ruled}
\newcolumntype{?}{!{\vrule width 1pt}}
\usepackage{pgfplots}
\usepackage{filecontents}

\usepackage{mathtools}

\usepackage{etoolbox}

\makeatletter
\newcommand{\changeoperator}[1]{%
  \csletcs{#1@saved}{#1@}%
  \csdef{#1@}{\changed@operator{#1}}%
}
\newcommand{\changed@operator}[1]{%
  \mathop{%
    \mathchoice{\textstyle\csuse{#1@saved}}
               {\csuse{#1@saved}}
               {\csuse{#1@saved}}
               {\csuse{#1@saved}}%
  }
}
\makeatother

\changeoperator{sum}
\changeoperator{prod}

\usepackage{tikz,tabularx}
\usetikzlibrary{decorations.pathreplacing}
\tikzstyle{every picture}+=[remember picture]

\usetikzlibrary{shapes,arrows,shadows}
\usetikzlibrary{patterns}
\usetikzlibrary{shapes.geometric}
\usetikzlibrary{arrows}

\usepackage{tikz,tabularx}
\usetikzlibrary{decorations.pathreplacing}
\tikzstyle{every picture}+=[remember picture]

\usetikzlibrary{decorations.markings}

\usepackage{multicol}
\usepackage{pgfplots}
\pgfplotsset{major grid style={densely dotted,black!70}}
\pgfplotsset{compat=1.9}
\usetikzlibrary{arrows}

\usepackage{amsthm}

\newtheorem{remark}{Remark}
\newtheorem{definition}{Definition}
\newtheorem{theorem}{Theorem}
\newtheorem{lemma}{Lemma}
\newtheorem{proposition}{Proposition}

\newcommand{\tcr}[1]{\textcolor{red}{#1}}

\usepackage{tikz,tabularx}
\usetikzlibrary{decorations.pathreplacing}
\usetikzlibrary{shapes.geometric, arrows.meta, fit, positioning}
\tikzstyle{every picture}+=[remember picture]

\usetikzlibrary{shapes,arrows,shadows}
\usepgflibrary{patterns}
\usetikzlibrary{patterns}
\usetikzlibrary{shapes.geometric}
\usetikzlibrary{arrows}
\usetikzlibrary{positioning}

\usepackage{tikz,tabularx}
\usetikzlibrary{decorations.pathreplacing}
\tikzstyle{every picture}+=[remember picture]

\usetikzlibrary{decorations.markings}

\usepackage{multicol}
\usepackage{pgfplots}
\title{Fundamental Limits of Covert queueing under Busy-Period Observations}

\author{%
  \IEEEauthorblockN{Anumula Venkata Sai Sree Sahithi\IEEEauthorrefmark{1}, Arti Yardi\IEEEauthorrefmark{1}, and Jithin Ravi\IEEEauthorrefmark{2}}
  \IEEEauthorblockA{\IEEEauthorrefmark{1}%
  International Institute of Information Technology Hyderabad, Hyderabad, India}
  \IEEEauthorblockA{\IEEEauthorrefmark{2}%
  Indian Institute of Technology Kharagpur, Kharagpur, India}
             \IEEEauthorblockA{Emails:sahithi.anumula@research.iiit.ac.in, arti.yardi@iiit.ac.in, jithin@ece.iitkgp.ac.in}
}

\begin{document}
\maketitle

\begin{abstract}
Consider a queueing system operated by Bob, which is intended to serve only jobs generated by a legitimate user, Willie. Willie's jobs arrive according to a Poisson process with rate $\lambda_w$, and the service times are exponentially distributed with rate $\mu$. Bob seeks to covertly serve additional jobs arriving according to an independent Poisson process with rate $\lambda_b$, without Willie detecting their presence. As the legitimate user, Willie is provided with certain statistics of each busy period by Bob. In particular, we consider two statistics: (i) the number of jobs served during each busy period and (ii) the duration of each busy period. We establish that, under either statistic, the maximum expected number of additional jobs that Bob can covertly serve obeys a square-root law. Specifically, over $n$ busy periods, Bob can covertly serve additional jobs in the order of  $\sqrt{n}$ whereas any scaling which is above $\sqrt{n}$  enables Willie to asymptotically detect the presence of these additional jobs with vanishing error probability. Furthermore, we characterize the scaling constant associated with the square-root law, which we refer to as the covert queueing constant, for each of the two statistics.
We compare the scaling constants corresponding to the two statistics and argue that the scaling constant under statistic (i) is no greater than that under statistic (ii).
\end{abstract}


\section{Introduction}

Consider a scenario in which two parties, Alice and Bob, collaborate to accomplish a common objective, while some of their actions or observable statistics are monitored by a warden, Willie. Alice and Bob wish to achieve their objective without Willie detecting their collaboration. Such problems broadly fall under the framework of \emph{stealth}. When the objective is specifically to communicate information without revealing the very presence of communication, the problem is commonly referred to as \emph{covert communication}.

In a standard covert communication setting, a transmitter Alice communicates with a receiver Bob over a noisy channel, while a warden Willie observes the transmission through another channel and attempts to determine whether Alice is transmitting. The fundamental question is how much information Alice can reliably communicate to Bob while keeping the transmission undetectable to Willie. A seminal result of Bash \textit{et al.}~\cite{Bash_sq_root_law_2012} established the \emph{square-root law}: over $n$ channel uses, Alice can covertly transmit on the order of $\sqrt{n}$ bits, whereas transmitting asymptotically more than $\sqrt{n}$ bits makes reliable detection by Willie possible. Following this seminal work, covert communication has been extensively studied under a variety of communication models \cite{Bash_sq_root_journal_2013, Bash2014, Bullock2020, Goeckel2016, Che2013,Che2014, Bloch2016, Tahmasbi2020Learning, Goeckel2017ActiveSensing, Sobers2017ContinuousTime, Zhang2021AdversariallyJammed, Wang2016}. 

In this paper, we investigate an analogous stealth problem in a queueing system. Consider a queue operated by Bob that is intended to serve only jobs generated by Willie. Bob, however, colludes with Alice and wishes to additionally serve Alice's jobs, which we refer to as \emph{non-Willie jobs}, without Willie being able to detect the presence of extra jobs being served. Specifically, we consider an $M/M/1$ queue in which Willie generates jobs according to a Poisson process with arrival rate $\lambda_w$, and the server operates at service rate $\mu$. To monitor the usage of the queue, Willie is provided with certain statistics of the resulting busy periods. The fundamental question is then: how many non-Willie jobs can Bob serve while keeping their presence undetectable from the statistics observed by Willie? We refer to this problem as \emph{covert queueing}.

Several works have studied covert insertion in packet and queueing systems \cite{Mukherjee2016, Ghassami2018, Ghassami2018RR}. Soltani \textit{et al.}~\cite{Soltani2015} studied a Poisson packet channel
in which Bob inserts packets while Willie observes packet timings and
showed that the number of covert packets follows square-root scaling.
This was further extended by Soltani \textit{et al.}~\cite{Towsley_packet_insertion_2020},
where general renewal inter arrival times and established
the same scaling. Covert insertion has also been studied in queueing
systems. Jiang \textit{et al.}~\cite{Towsley_cycle_stealing_2020}
considered a single FIFO server and showed that the number of covert
jobs depends on the service rates of Bob's and Willie's jobs, with
square-root scaling under certain conditions. Most closely related
to our work, Yardi \textit{et al.}~\cite{Comm_letter_paper}
considered a queue in which Willie observes either the number of jobs
served or the busy period duration. They showed $O(\sqrt{n})$
scaling for the former, but only $O(1)$ expected covert service time
for the latter. Yardi \textit{et al.}~\cite{Our_ITW_MC_Queueing}
also considered a Markovian statistic and obtained square-root
scaling.

Motivated by these works, we consider a queueing system in which covert or illegitimate traffic is deliberately introduced alongside the legitimate traffic generated by Willie. Bob provides Willie with one of the following two statistics for each busy period:
\begin{enumerate}
\item \emph{Statistic-1}: the number of jobs served during the $i$-th busy period, denoted by $X_i$;
\item \emph{Statistic-2}: the duration of the $i$-th busy period, denoted by $Y_i$.
\end{enumerate}
Based on the statistics observed over $n$ busy periods, Willie seeks to determine whether Bob is serving non-Willie jobs. This problem was first introduced in~\cite{Comm_letter_paper}, where the analysis was limited to achievability results for the scaling of the number of additional jobs under the two statistics. Specifically, an achievability result of order $\sqrt{n}$ was established for Statistic-1, whereas only a constant-order achievability result was obtained for Statistic-2. In this work, we establish the fundamental scaling law for both statistics by proving that the square-root law holds in both cases. In this work, we characterize the fundamental scaling law for both statistics. We show that the square-root law holds in both cases: Bob can covertly serve on the order of $\sqrt{n}$ non-Willie jobs, whereas serving asymptotically more than $\sqrt{n}$ non-Willie jobs enables Willie to reliably detect their presence.

Having established the square-root scaling, we next investigate the corresponding scaling constants. In conventional covert communication, the optimal constant associated with the square-root law is commonly referred to as the \emph{covert capacity}. Such constants were characterized for classes of discrete memoryless channels and the Gaussian channel in~\cite{Wang2016}. Motivated by this notion, we characterize the corresponding square-root scaling constants for the two queueing statistics considered here. Specifically, Theorem~\ref{Thm:KX} and Theorem~\ref{Thm:KY} characterize the scaling constants associated with Statistic-1 and Statistic-2, respectively.

The rest of the paper is organized as follows. \cref{sec:system_model} presents the preliminaries and problem formulation. \cref{sec:Main_results} establishes the main results for the two statistics, including their square-root scaling laws and covert queueing constants. \cref{sec:Discussion} compares the two statistics and discuss about covert capacity. \cref{sec:Conclusion}
concludes the paper and discusses future directions. Detailed proofs
are provided in the appendices.

\section{Preliminaries and Problem formulation}
\label{sec:system_model}

\subsection{Preliminaries}
\color{black}

\subsubsection{Preliminaries on the $M/M/1$ Queue}

We consider an $M/M/1$ queue with arrival rate $\lambda$ and service rate $\mu$, where $0<\lambda<\mu$. The interarrival times are exponentially distributed with rate $\lambda$, while the service times are exponentially distributed with rate $\mu$. The condition $\lambda<\mu$ ensures the stability of the queue. The server utilization factor, denoted by $\rho$, is defined as the ratio of the arrival rate to the service rate, i.e.,
\begin{align}
\rho \triangleq \frac{\lambda}{\mu}.\nonumber
\end{align}

We divide the operation of the queue into successive \emph{sessions}, where each session corresponds to a busy period of the queue. A session begins when a job arrives while the queue is empty. Following this arrival, the server remains busy as long as there is at least one job in the system. The session ends when the service of a job is completed and the system becomes empty again. We denote the $i$-th session by $S_i$.

Let $X_i$ denote the total number of jobs served during session $S_i$. By the exponential interarrival and service times, the random variables ${X_i}$ are independent and identically distributed (i.i.d.). We denote their common distribution by that of the random variable $X$. For an $M/M/1$ queue, the probability mass function of $X$ is given by~\cite{Queueing_systems_book}
\begin{equation}
\label{eq:number_stat}
f_X(x)
=
\frac{1}{x}
\binom{2x-2}{x-1}
\frac{\rho^{x-1}}{(1+\rho)^{2x-1}},
\qquad x=1,2,\ldots,
\end{equation}
where $\binom{a}{b}$ denotes the binomial coefficient. The expected number of jobs served during a session is
\begin{equation}
\mathbb{E}[X]
=
\frac{\mu}{\mu-\lambda}.
\label{eq:expectation_stat1}
\end{equation}

Let $Y_i$ denote the duration of session $S_i$. The random variables ${Y_i}$ are also i.i.d., and we denote their common distribution by that of the random variable $Y$. For an $M/M/1$ queue, the probability density function of $Y$ is given by~\cite{Queueing_systems_book}
\begin{equation}
\label{eq:length_stat}
g_Y(y)
=
e^{-(\lambda+\mu)y}
\sqrt{\frac{\mu}{\lambda}}
\frac{I_1\left(2y\sqrt{\lambda\mu}\right)}{y},
\qquad y>0,
\end{equation}
where $I_1(\cdot)$ denotes the modified Bessel function of the first kind of order one. The expected duration of a session is
\begin{equation}
\mathbb{E}[Y]
=
\frac{1}{\mu-\lambda}.
\end{equation}

\subsubsection{Order Notations}
Let $f(x)$ and $g(x)$ be two real valued functions, and let $a \in \mathbb{R}\cup
\{\infty\}$.  We write
$f(x)=O(g(x))$ if $\limsup_{x \to a} \frac{|f(x)|}{|g(x)|} < \infty$. Similarly, we write  $f(x)=\Omega(g(x))$ if $\liminf_{x \to a} \frac{|f(x)|}{|g(x)|} > 0$. Furthermore,  $f(x)=\Theta(g(x))$ if
$f(x)$ is both $O(g(x))$ and $\Omega(g(x))$. We write $f(x)=o(g(x))$ if
$\lim_{x \to a} \frac{|f(x)|}{|g(x)|} =0$. Finally, we write $f(x)=\omega(g(x))$ if
$\liminf_{x \to a} \frac{|f(x)|}{|g(x)|} = \infty$.

\subsection{Problem formulation}

Consider a queueing system in which Bob operates a server with service rate $\mu$. A legitimate user, Willie, submits jobs to the server according to a Poisson process with arrival rate $0<\lambda_w<\mu$. Bob wishes to covertly introduce an additional, independent Poisson stream of non-Willie jobs with arrival rate $\lambda_b\geq 0$. Consequently, when Bob inserts additional jobs, the aggregate arrival process is Poisson with rate $\lambda_w+\lambda_b$. We assume throughout that
\begin{align}
\lambda_w+\lambda_b<\mu,\nonumber
\end{align}
so that the queue remains stable.

As part of their agreement, after the completion of $n$ sessions, Bob provides Willie with a usage statistic $U_i$ corresponding to each session $S_i$. Based on the observations
\begin{align}
U^n \triangleq (U_1,U_2,\ldots,U_n),\nonumber
\end{align}
Willie attempts to determine whether Bob has introduced an additional Poisson stream of jobs. Thus, Willie performs the following binary hypothesis test:
\begin{equation}
\begin{aligned}
H_0 &: \lambda=\lambda_w,\\
H_1 &: \lambda=\lambda_w+\lambda_b.
\end{aligned}
\label{Eq:Hypothesis_setup}
\end{equation}

Let $\alpha_n$ and $\beta_n$ denote the false-alarm and missed-detection probabilities, respectively, associated with Willie's detection rule, i.e.,
\begin{align}
\alpha_n
&\triangleq
\Pr(\text{decide }H_1\mid H_0\text{ is true}), \notag \\
\beta_n
&\triangleq
\Pr(\text{decide }H_0\mid H_1\text{ is true}).\nonumber
\end{align}
Let $f_j(u)$, $j=0,1$, denote the distribution of the statistic $U$ under hypothesis $H_j$. Since the statistics $U_1,U_2,\ldots,U_n$ are assumed to be i.i.d., their joint distribution under $H_j$ is given by
\begin{equation}
f_j^n(u^n)
=
\prod_{i=1}^{n} f_j(u_i),
\qquad j\in{0,1}.\nonumber
\end{equation}
For the optimal detection rule, it follows that~\cite[Thm.~13.1.1]{Lehmann_Book}
\begin{align}
\alpha_n+\beta_n
=
1-\operatorname{TV}
\left(
f_0^n,f_1^n
\right),\nonumber
\end{align}
where $\operatorname{TV}(P,Q)$ denotes the total variation distance between the distributions $P$ and $Q$.

For a given $\epsilon\in(0,1/2)$, Bob is said to achieve $\epsilon$-covert queueing if
\begin{align}
\alpha_n+\beta_n
\geq
1-\epsilon.\nonumber
\end{align}
By Pinsker's inequality~\cite[Ch.~11]{Thomas_Cover_Book},
\begin{align}
\operatorname{TV}(P,Q)
\leq
\sqrt{\frac{1}{2}D_{\mathrm{KL}}(P\|Q)},\nonumber
\end{align}
where $D_{\mathrm{KL}}(P\|Q)$ denotes the Kullback--Leibler (KL) divergence\footnote{Here, $\log$ denotes the natural logarithm. For continuous probability distributions, the summation is replaced by the corresponding integral over the probability density functions.} between $P$ and $Q$, and is defined 
\begin{align}
D_{\mathrm{KL}}(P\|Q)
\triangleq
\sum_x P(x) \log \frac{P(x)}{Q(x)}.\nonumber
\end{align}

Therefore, if
\begin{align}
D_{\mathrm{KL}}
\left(
f_0^n\|f_1^n
\right)
\leq\delta,
\qquad 0<\delta<1,\nonumber
\end{align}
then
\begin{align}
\alpha_n+\beta_n
\geq
1-\sqrt{\frac{\delta}{2}}.\nonumber
\end{align}
\begin{definition}
\label{covert_defibition}
Bob's additional arrival process is said to be \emph{$\delta$-covert with respect to the statistic $U$} if
\begin{align}
\label{eq:delta_covert}
    \limsup_{n\to\infty}
    D_{\mathrm{KL}}\left(f_0^n\|f_1^n\right)
    \leq\delta,
    \qquad 0<\delta<1.
\end{align}
\end{definition}

For a sequence of arrival rates
$\{\lambda_b(n)\}_{n\geq 1}$, let $E_{b,U}(n)$
denote the expected number of additional jobs inserted
by Bob over $n$ sessions, given by
\begin{align}
\label{eq:expected_jobs}
E_{b,U}(n)
&=n\frac{\lambda_b(n)}
{\lambda_w+\lambda_b(n)}\mathbb{E}[X],
\end{align}
where the expectation is taken with respect to
the distribution $f_X(x)$ defined in
\eqref{eq:number_stat}.

Our first objective is to characterize the maximum
asymptotic scaling of $E_{b,U}(n)$ achievable under
the covertness constraint in \eqref{eq:delta_covert}.
We denote this optimal scaling by $E_{b,U}^*(n)$.
More precisely, we say that
\begin{align}
E_{b,U}^*(n)=\Theta(f(n))
\end{align}
if the following two conditions hold:

\begin{enumerate}
    \item There exists a sequence of
    arrival rates $\{\lambda_b(n)\}_{n\geq 1}$
    satisfying the covertness constraint such that
    \begin{align}
        E_{b,U}(n)=\Omega(f(n)).
    \end{align}

    \item For every sequence of
    arrival rates $\{\lambda_b(n)\}_{n\geq 1}$
    satisfying the covertness constraint,
    \begin{align}
        E_{b,U}(n)=O(f(n)).
    \end{align}
\end{enumerate}

In particular, we consider the statistics $U=X$ and $U=Y$ defined in \eqref{eq:number_stat} and in \eqref{eq:length_stat}, respectively\footnote{To simplify the notation, we denote the distribution of $U=X$ by $f_X(x)$ and that of $U=Y$ by $g_Y(y)$. When there is no ambiguity, we omit the subscripts and simply write $f(x)$ and $g(y)$, respectively.}. When $U=X$, Bob reveals the number of jobs served in each session, i.e.,
\begin{align}
X^n=(X_1,X_2,\ldots,X_n),\nonumber
\end{align}
whereas when $U=Y$, Bob reveals the duration of each session, i.e.,
\begin{align}
Y^n=(Y_1,Y_2,\ldots,Y_n).\nonumber
\end{align}

In Propositions~\ref{Proposition:stat1_achievability} and~\ref{Proposition:stat2}, we show that, for both $U=X$ and $U=Y$, the maximum number of additional jobs that can be covertly inserted obeys the square-root scaling law
\begin{align}
E_{b,U}^{*}(n)=\Theta(\sqrt{n}),
\qquad U\in\{X,Y\}.
\end{align}
A similar square-root scaling law arises in the covert communication problem studied \cite{Wang2016}, where the constant associated with the $\sqrt{n}$ scaling is characterized and referred to as the \emph{covert capacity}. Motivated by this result, and since the maximum number of additional jobs in our setting also follows a square-root scaling law, we are interested in characterizing the corresponding scaling constant. Specifically, we define the \emph{covert queueing constant} associated with the statistic $U$ as
\begin{equation}
K_U
\triangleq
\liminf_{n\to\infty}
\frac{E_{b,U}^{*}(n)}{\sqrt{n\delta}},
\qquad U\in{\{X,Y}\},
\label{Eqn_KX_Definition}
\end{equation}
where $\delta$ denotes the covertness constraint. Our goal in this paper is to characterize the covert queueing constants $K_X$ and $K_Y$.

\section{Main results}
\label{sec:Main_results}

In this section, we present our results for the statistic $U=X$ in
\cref{sec:statistic1} and for the statistic $U=Y$ in
\cref{sec:statistic2}. We first establish the scaling relation between $E_{b,U}(n)$ and $\lambda_b(n)$ in Proposition~\ref{prop:ebn_lambda_b}.

\begin{proposition}
\label{prop:ebn_lambda_b}
Suppose that Bob employs a sequence of arrival rates
$\{\lambda_b(n)\}_{n\geq 1}$ such that the queue is uniformly stable, i.e.,
there exists an $\eta>0$ and $n_0$ for which
\begin{align}
    \lambda_w+\lambda_b(n) \leq \mu-\eta,
    \qquad  n \geq n_0.\nonumber
\end{align}
Then, the expected number of additional jobs inserted by Bob satisfies
\begin{equation}
\label{eq:Eb_equiv}
    E_{b,U}(n)
    =
    n\Theta\bigl(\lambda_b(n)\bigr),
    \qquad U\in\{X,Y\}.\nonumber
\end{equation}
Equivalently,
\begin{align}
    \lambda_b(n)
    =
    \frac{1}{n}\Theta\bigl(E_{b,U}(n)\bigr).
\end{align}
\end{proposition}

\begin{proof}
 The proof is elementary and is given in Appendix~\ref{app:D}. 
\end{proof}

 We further give a remark on the Taylor expansion of the KL divergence.
\begin{remark}
Suppose that $\{P_{\kappa}:\kappa\in\mathbb{R}^{+}\}$ is a parametric family of distributions, and consider two nearby distributions $P_{\kappa}$ and $P_{\kappa'}$, where
\begin{align}
    \kappa'=\kappa+\gamma,
    \qquad \gamma>0.\nonumber
\end{align}
Assuming smoothness with respect to $\gamma$, a Taylor expansion of
$D_{\mathrm{KL}}(P_{\kappa}\|P_{\kappa+\gamma})$
around $\gamma=0$ yields
\begin{align}
D_{\mathrm{KL}}(P_{\kappa}\|P_{\kappa+\gamma})
&=
D_{\mathrm{KL}}(P_{\kappa}\|P_{\kappa})
\nonumber\\
&\quad+
\gamma
\left.
\frac{\partial}{\partial\gamma}
D_{\mathrm{KL}}(P_{\kappa}\|P_{\kappa+\gamma})
\right|_{\gamma=0}
\nonumber\\
&\quad+
\frac{\gamma^2}{2}
\left.
\frac{\partial^2}{\partial\gamma^2}
D_{\mathrm{KL}}(P_{\kappa}\|P_{\kappa+\gamma})
\right|_{\gamma=0}
+
O(\gamma^3)
\nonumber\\
&=
\frac{\gamma^2}{2}
\left.
\frac{\partial^2}{\partial\gamma^2}
D_{\mathrm{KL}}(P_{\kappa}\|P_{\kappa+\gamma})
\right|_{\gamma=0}
+
O(\gamma^3).
\label{eq:Taylor_KL}
\end{align}
The last equality follows from the nonnegativity of the KL divergence. In particular,
\begin{align}
D_{\mathrm{KL}}(P_{\kappa}\|P_{\kappa+\gamma})\geq 0,\nonumber
\end{align}
with equality at $\gamma=0$, since
$P_{\kappa+\gamma}=P_{\kappa}$ when $\gamma=0$. Thus, $\gamma=0$ is a minimum of
$D_{\mathrm{KL}}(P_{\kappa}\|P_{\kappa+\gamma})$. Consequently,
\begin{align}
D_{\mathrm{KL}}(P_{\kappa}\|P_{\kappa})=0,
\qquad
\left.
\frac{\partial}{\partial\gamma}
D_{\mathrm{KL}}(P_{\kappa}\|P_{\kappa+\gamma})
\right|_{\gamma=0}
=0,\nonumber
\end{align}
which establishes \eqref{eq:Taylor_KL}.
\end{remark}

\subsection{Results on Statistic $U=X$}
\label{sec:statistic1}

When the revealed statistic $U$ is the number of jobs served during a session, i.e., $U=X$, its distribution is given by \eqref{eq:number_stat}. Denoting the distribution of $X$ under hypothesis $H_j$, $j = 0,1$, by $f_j(x)$, we have
\begin{equation}
\label{eq:number_stat_res}
f_j(x)
=
\frac{1}{x}
\binom{2x-2}{x-1}
\frac{\rho_j^{x-1}}{(1+\rho_j)^{2x-1}},
\qquad x=1,2,\ldots,
\end{equation}
where
\begin{align}
\label{eq:rhos}
\rho_0=\frac{\lambda_w}{\mu},
\qquad
\rho_1=\frac{\lambda_w+\lambda_b}{\mu}.
\end{align}

It was shown  in \cite[Theorem~1]{Comm_letter_paper} that, when $U=X$, Bob can covertly insert jobs with a scaling of order $\sqrt{n}$. Hence, we have
\begin{align}
\label{eq:scale_Eb}
E_{b,X}(n)=\Omega(\sqrt{n}).
\end{align}
The following proposition establishes that no scaling strictly larger than $\sqrt{n}$ is possible.

\begin{proposition}
\label{Proposition:stat1_achievability}
Suppose that the statistic provided by server Bob to Willie for each session $S_i$ is the number of jobs served, $X_i$. If
\begin{align}
E_{b,X}(n)=\omega(\sqrt{n}),
\end{align}
then there exists a hypothesis test for Willie such that 
\begin{align}
\alpha_n+\beta_n \to 0
\end{align}
as $n\to\infty$. Consequently, Bob cannot covertly insert jobs at any scaling strictly larger than $\sqrt{n}$.
\end{proposition}

\begin{proof}
See Appendix~\ref{app:A}.
\end{proof}

Proposition~\ref{Proposition:stat1_achievability}, together with
\eqref{eq:scale_Eb}, establishes that the maximum achievable scaling of
$E_b(n)$  under statistic $U=X$ is $\sqrt{n}$, i.e.,
\begin{align}
\label{eq:Eb_X_max}
E_{b,X}^{*}(n)=\Theta(\sqrt{n}).
\end{align}
It follows from \eqref{eq:Eb_X_max} that the scaling constant $K_X$ defined in
\eqref{Eqn_KX_Definition} is finite and strictly positive. We provide an explicit characterization of $K_X$ in Theorem~\ref{Thm:KX}. The following lemma will be useful in proving Theorem~\ref{Thm:KX}.

\begin{lemma}
\label{lem:stat1_kl}
The distribution $f_j(x)$ of $X$ under hypothesis $H_j,j=0,1$, is parameterized by $\rho_j$ as defined in \eqref{eq:rhos}. In particular,
\begin{align}
\rho_1=\rho_0+\frac{\lambda_b}{\mu}.\nonumber
\end{align}
Applying the Taylor expansion in \eqref{eq:Taylor_KL}, we obtain
\begin{align}
D_{\mathrm{KL}}(f_0\|f_1)
&=
\frac{1}{2}
\frac{\mu}
{\lambda_w(\mu^2-\lambda_w^2)}
\lambda_b^2
+
O(\lambda_b^3).
\label{eq:statX_KL_taylor}
\end{align}
\end{lemma}
\begin{proof}
See Appendix~\ref{app:E}.
\end{proof}

\begin{theorem}
\label{Thm:KX}
Consider an $M/M/1$ queue with service rate $\mu$
and Willie job arrival rate $\lambda_w$, where
$0<\lambda_w<\mu$. Suppose that the statistic
provided to Willie for the $i$-th session is the
number of jobs $X_i$ served during that session.
Then, under the asymptotic KL-divergence
covertness constraint, the scaling constant
$K_X$ defined in \cref{Eqn_KX_Definition} is
given by
\begin{equation}
K_X
=
\sqrt{2}
\sqrt{
\frac{\mu(\mu+\lambda_w)}
{\lambda_w(\mu-\lambda_w)}
}.
\end{equation}
\end{theorem}

\begin{proof}
First let us define
\begin{align}
c_X
\triangleq
\frac{\mu}
{2\lambda_w(\mu^2-\lambda_w^2)}.
\end{align}
From \cref{eq:statX_KL_taylor}, as
$\lambda_b\to 0$,
\begin{align}
D_{\mathrm{KL}}(f_0\|f_1)
=
c_X\lambda_b^2+o(\lambda_b^2).
\label{eq:KX_KL_expansion}
\end{align}
Moreover, the expected number of jobs inserted
by Bob over $n$ sessions is
\begin{align}
E_{b,X}(n)
&=
n\frac{\lambda_b(n)}
{\lambda_w+\lambda_b(n)}
\mathbb{E}[X]
\nonumber\\
&=
\frac{n\mu\lambda_b(n)}
{(\lambda_w+\lambda_b(n))
(\mu-\lambda_w-\lambda_b(n))}.
\label{eq:KX_Eb}
\end{align}

\emph{Achievability:}
Choose the insertion rate
\begin{align}
\lambda_b(n)
=
\sqrt{\frac{\delta}{nc_X}}.
\label{eq:KX_ach_rate}
\end{align}
Since the session statistics are i.i.d.,
\cref{eq:KX_KL_expansion} yields
\begin{align}
D_{\mathrm{KL}}
(f_0^{n}\|f_1^{n})
&=
nD_{\mathrm{KL}}(f_0\|f_1)
\nonumber\\
&=
nc_X\lambda_b(n)^2+o(1)
\nonumber\\
&=
\delta+o(1).
\end{align}
Thus, the asymptotic KL-divergence
covertness constraint is satisfied.

Furthermore, substituting
\cref{eq:KX_ach_rate} into \cref{eq:KX_Eb}
gives
\begin{align}
\lim_{n\to\infty}
\frac{E_{b,X}(n)}{\sqrt{n\delta}}
&=
\frac{\mu}
{\lambda_w(\mu-\lambda_w)\sqrt{c_X}}
\nonumber\\
&=
\sqrt{2}
\sqrt{
\frac{\mu(\mu+\lambda_w)}
{\lambda_w(\mu-\lambda_w)}
}.
\end{align}
This establishes the achievability bound.

\emph{Converse:}
Consider any sequence of insertion rates
$\{\lambda_b(n)\}_{n\geq 1}$ satisfying
the asymptotic KL-divergence covertness
constraint
\begin{align}
\limsup_{n\to\infty}
D_{\mathrm{KL}}
(f_0^{n}\|f_1^{n})
\leq \delta.
\label{eq:KX_KL_constraint}
\end{align}
Since the divergence vanishes only when
$\lambda_b=0$ and is bounded away from zero
for insertion rates bounded away from zero,
\cref{eq:KX_KL_constraint} implies
$\lambda_b(n)\to 0$.

Consequently, from
\cref{eq:KX_KL_expansion}, we obtain
\begin{align}
\limsup_{n\to\infty}
n\lambda_b(n)^2
\leq
\frac{\delta}{c_X}.
\label{eq:KX_conv_rate}
\end{align}
Using \cref{eq:KX_Eb} and
\cref{eq:KX_conv_rate}, it follows that
\begin{align}
\limsup_{n\to\infty}
\frac{E_{b,X}(n)}{\sqrt{n\delta}}
&\leq
\frac{\mu}
{\lambda_w(\mu-\lambda_w)\sqrt{c_X}}
\nonumber\\
&=
\sqrt{2}
\sqrt{
\frac{\mu(\mu+\lambda_w)}
{\lambda_w(\mu-\lambda_w)}
}.
\end{align}
Combining the achievability and converse
bounds completes the proof.
\end{proof}

\subsection{Results on Statistic $U=Y$}
\label{sec:statistic2}

When the revealed statistic $U$ is the duration of each session, i.e., $U=Y$, its distribution is given by \cref{eq:length_stat}. Denoting the distribution of $Y$ under hypothesis $H_j$, $j = 0,1$, by $g_j(y)$, we have
\begin{equation}
\label{eq:length_stat_res}
g_j(y)
=
e^{-(\xi_j+\mu)y}
\sqrt{\frac{\mu}{\xi_j}}
\frac{I_1\left(2y\sqrt{\xi_j\mu}\right)}{y},
\qquad y>0,
\end{equation}
where $I_1(\cdot)$
denotes the first-order modified Bessel function, and
\begin{align}
\label{eq:lambda_under_hj}
\xi_0={\lambda_w},
\qquad
\xi_1={\lambda_w+\lambda_b}. 
\end{align}

In \Cref{Proposition:stat2}, we show that $E_{b,Y}(n)$ obeys the square-root law, analogous to $E_{b,X}(n)$. To this end, we first derive a Taylor series approximation of $D_{\mathrm{KL}}(g_0|g_1)$ in \Cref{lem:Tayor_D_KL_Stat2}. We point out that, due to the presence of the Bessel function term in $g_j(y)$, $j=0,1$, deriving this Taylor expansion is non-trivial. Therefore, the detailed derivation is provided in Appendix~\ref{app:C}.

\begin{lemma}
\label{lem:Tayor_D_KL_Stat2}

For the densities $g_0(y)$ and $g_1(y)$ in \cref{eq:length_stat_res}, the second-order Taylor expansion of $D_{\mathrm{KL}}(g_0\|g_1)$ is given by
\begin{equation}
D_{\mathrm{KL}}(g_0\|g_1)
=
C(\lambda_w, \mu) \lambda_b^2+o(\lambda_b^2),
\label{Eqn:Tayor_D_KL_Stat2}
\end{equation}
where
\begin{equation}
C(\lambda_w, \mu) \triangleq
\frac{1}{2}
\left[
-\frac{1}{\lambda_w^2}
-\frac{2\mu^2}{\lambda_w(\mu-\lambda_w)^3}
+\frac{\mu^{3/2}}{\lambda_w^{3/2}}J(\lambda_w, \mu)
\right]
\label{Eq:C_lambda_mu_remark}
\end{equation}
and
\begin{align}
J(\lambda_w, \mu) \triangleq \int_0^\infty
y e^{-(\mu+\lambda_w)y}
\frac{I_0^2(2y\sqrt{\mu \lambda_w})}{I_1(2y\sqrt{\mu \lambda_w})}
\,dy.
\label{Eq:J_lambda_mu_remark}
\end{align}
Note that $C(\lambda_w, \mu)$ and $J(\lambda_w, \mu)$ are functions of $\lambda_w$ and $\mu$ and can be evaluated numerically.
\end{lemma}
\begin{proof}
    See Appendix~\ref{app:C}.
\end{proof}

The following proposition establishes the square-root law for $E_{b,Y}^{*}(n)$.
\begin{proposition}
\label{Proposition:stat2}
Suppose that the statistic provided by server Bob to Willie for each session $S_i$ is the length of each session $Y_i$. Then, 
\begin{align}
  E_{b,Y}^{*}(n)=\Theta(\sqrt{n})  
\end{align}
In particular, Bob can covertly insert an expected number
    of non-Willie jobs satisfying
    \begin{equation}
    E_{b,Y}(n)=\Omega(\sqrt{n}).
    \label{eq:eb_y_stat2}
    \end{equation}
Further, if Bob inserts an expected number of
    jobs satisfying
    \begin{equation}
    E_{b,Y}(n)=\omega(\sqrt{n}),
    \end{equation}
    then there exists a hypothesis test for Willie such that
    $\alpha_n+\beta_n\to 0$
    as $n\to\infty$. 
\end{proposition}
\begin{proof}
See Appendix~\ref{app:B}.
\end{proof}

\begin{theorem}
\label{Thm:KY}
Consider an $M/M/1$ queue with service rate $\mu$ and Willie job arrival rate $\lambda_w$, where $0<\lambda_w<\mu$. Suppose that the statistic provided to Willie for the $i$-th session is the length of each session $Y_i$. Then, the scaling constant $K_Y$, defined in \eqref{Eqn_KX_Definition}, is given by
\begin{equation}
\label{eq:ky_eq}
K_Y
=
\frac{\mu}
{\lambda_w(\mu-\lambda_w)\sqrt{C(\lambda_w, \mu)}},
\end{equation}
where $C(\lambda_w, \mu)$ is a constant defined in \cref{Eq:C_lambda_mu_remark}.
\end{theorem}

\begin{proof}
The proof follows along similar lines to that of Theorem~\ref{Thm:KX}. The detailed steps are provided in Appendix~\ref{app:proof_thm_KY}.
\end{proof}

\section{Discussion}
\label{sec:Discussion}

\begin{figure*}[ht]
\centering

\begin{minipage}{0.45\textwidth}
    \centering
    \begin{tikzpicture}
\begin{axis}[
    width= 7.5cm,
    legend style={at={(0.3,0.25)},anchor=west},
    ylabel={\footnotesize{$K_X, K_Y$}},
    xlabel={\footnotesize{$\lambda_w$}},
    xmin=0.05,
    xmax=4.95,
    grid=both,
    xtick={0.5,1,1.5,2,2.5,3,3.5,4,4.5},
    xticklabels={
        \scriptsize{0.5},
        \scriptsize{1},
        \scriptsize{~},
        \scriptsize{2},
        \scriptsize{~},
        \scriptsize{3},
        \scriptsize{~},
        \scriptsize{4},
        \scriptsize{~}
    }
]

\addplot[
    color=green,
    very thick,
    mark=o,
    mark options={solid},
    mark size=1.5pt
]
table[
    x=lambda_w,
    y=KX,
    col sep=comma
]{figures/kx_vs_muky.csv};
\addlegendentry{\scriptsize{$K_X$}}

\addplot[
    color=red,
    very thick,
    dashed,
    mark=x,
    mark size=1.5pt
]
table[
    x=lambda_w,
    y=mu_KY,
    col sep=comma
]{figures/kx_vs_muky.csv};
\addlegendentry{\scriptsize{$K_Y$}}

\end{axis}
\end{tikzpicture}
\end{minipage}
\hfill
\begin{minipage}{0.45\textwidth}
    \centering
    \begin{tikzpicture}
\begin{axis}[
    width=7.5cm,
    legend style={at={(0.3,0.25)},anchor=west},
    ylabel={\footnotesize{$K_X, K_Y$}},
    xlabel={\footnotesize{$\lambda_w$}},
    xmin=1.03,
    xmax=3.97,
    grid=both,
    xtick={1,1.5,2,2.5,3,3.5},
    xticklabels={
        \scripAtsize{1},
        \scriptsize{~},
        \scriptsize{2},
        \scriptsize{~},
        \scriptsize{3},
        \scriptsize{~}
    }
]

\addplot[
    color=green,
    very thick,
    mark=o,
    mark options={solid},
    mark size=1.5pt
]
table[
    x=lambda_w,
    y=KX,
    col sep=comma
]{figures/kx_vs_muky_zoomed.csv};
\addlegendentry{\scriptsize{$K_X$}}

\addplot[
    color=red,
    very thick,
    dashed,
    mark=x,
    mark size=1.5pt
]
table[
    x=lambda_w,
    y=mu_KY,
    col sep=comma
]{figures/kx_vs_muky_zoomed.csv};
\addlegendentry{\scriptsize{$K_Y$}}

\end{axis}
\end{tikzpicture}
\end{minipage}

\caption{Comparison of $K_X$ and $K_Y$ for $\mu=5$. The right plot shows a zoomed view.}
\label{fig:kx_ky_comparison}
\end{figure*}
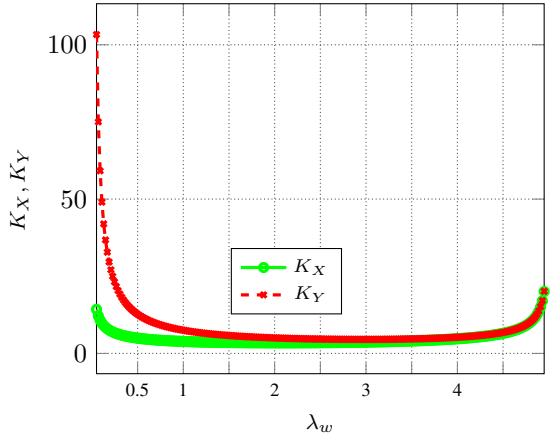

The scaling constants $K_X$ and $K_Y$, characterized in
\cref{Thm:KX,Thm:KY}, respectively, can be interpreted
as the \emph{covert queueing capacities} from Bob's
perspective when Willie observes statistic $X$ or $Y$.
Intuitively, a statistic that provides Willie with more
information about Bob's activity enables more effective
detection, thereby restricting the number of jobs Bob
can covertly insert while satisfying the covertness
constraint. Thus, comparing $K_X$ and $K_Y$ allows us
to assess the relative effectiveness of the two
statistics for detecting Bob's activity. We next show
that $X$ is more informative than $Y$ in terms of
KL divergence and consequently establish that
\begin{equation}
    K_X \leq K_Y.
\end{equation}

Recall that $X$ denotes the number of jobs served
during a busy period, whereas $Y$ denotes its duration.
In an $M/M/1$ queue with service rate $\mu$, the
service times are independent and identically
distributed (i.i.d.) exponential random variables
with rate $\mu$. Consequently, the busy-period
duration can be expressed as
\begin{equation}
    Y=\sum_{\ell=1}^{X} Z_\ell,
\end{equation}
where $\{Z_\ell\}_{\ell\geq 1}$ are i.i.d. random
variables following $\operatorname{Exp}(\mu)$.

Since the service rate $\mu$ is independent of
whether Bob inserts additional jobs, the conditional
distribution of $Y$ given $X=k$ is identical under
both hypotheses. In particular,
\begin{align}
g_{Y|X,H_0}(y|k)
&=g_{Y|X,H_1}(y|k) \nonumber\\
&=\operatorname{Erlang}(k,\mu),
\qquad y>0,
\end{align}
which is the density of an Erlang distribution
with shape parameter $k$ and rate parameter
$\mu$~\cite[Sec.~5.2.4]{Ross}.

Consider the binary hypothesis testing problem
in \cref{Eq:Hypothesis_setup}, and let
$\theta\in\{0,1\}$ denote the hypothesis index.
Since the conditional distribution of $Y$ given
$X$ does not depend on $\theta$, the random
variables form the Markov chain
\begin{equation}
    \theta \longrightarrow X \longrightarrow Y.
\end{equation}
Therefore, by the data processing inequality
for KL divergence~\cite[Thm.~2.15]
{polyanskiy2025information},
\begin{equation}
D_{\mathrm{KL}}(g_0\|g_1)
\leq
D_{\mathrm{KL}}(f_0\|f_1).
\label{Eq:DPI_KL}
\end{equation}
Thus, observing the number of jobs served
during a busy period is at least as informative
as observing its duration in terms of
KL divergence.

We now use \cref{Eq:DPI_KL} to compare the
corresponding covert queueing constants.
Fix $\delta>0$ and consider any sequence
of arrival rates $\{\lambda_b(n)\}_{n\geq1}$
satisfying the $\delta$-covertness constraint
for statistic $X$, namely,
\begin{align}
\limsup_{n\to\infty}
D_{\mathrm{KL}}
(f_0^{n}\|f_1^{n})
\leq \delta.
\end{align}
By \cref{Eq:DPI_KL}, the same sequence
necessarily satisfies
\begin{align}
\limsup_{n\to\infty}
D_{\mathrm{KL}}
(g_0^{\otimes n}\|g_1^{\otimes n})
\leq \delta.
\end{align}
Hence, every arrival-rate sequence that is
$\delta$-covert with respect to statistic $X$
is also $\delta$-covert with respect to
statistic $Y$.

Consequently, the set of admissible arrival-rate
sequences under statistic $X$ is contained in
the corresponding set under statistic $Y$.
Since the expected number of additional jobs
inserted by Bob is the same for a given
arrival-rate sequence, irrespective of which
statistic Willie observes, the optimal expected
number of covert jobs under statistic $Y$
cannot be smaller than that under statistic $X$.
Therefore, after normalization by
$\sqrt{n\delta}$ and taking the limit inferior
as in \cref{Eqn_KX_Definition}, we obtain
\begin{equation}
    K_X\leq K_Y.
\end{equation}

This inequality establishes that, under the
same covertness constraint, Bob can covertly
insert at least as many expected additional jobs
when Willie observes the busy-period duration
$Y$ as when Willie observes the number of
jobs $X$ served during each busy period.

The comparison is illustrated numerically in
\cref{fig:kx_ky_comparison} for $\mu=5$.
The figure confirms that $K_X\leq K_Y$
throughout the plotted range of $\lambda_w$.
Furthermore, both constants appear to diverge
as $\lambda_w\to0$ and as
$\lambda_w\to\mu$. An analytical characterization of the limiting
behavior of $K_X$ and $K_Y$ requires a
detailed investigation of the functions
$J(\lambda_w,\mu)$ and $C(\lambda_w,\mu)$
introduced in \cref{lem:Tayor_D_KL_Stat2}.
We leave this investigation for future work.

\section{Conclusion and future work}
\label{sec:Conclusion}

In this work, we investigated covert queuing in which Bob inserts an independent Poisson stream of jobs into a queue serving Willie's legitimate traffic. We considered two statistics available to Willie for detecting the presence of covert traffic: (i) the number of jobs served during a busy period $X$ and (ii) the busy-period length $Y$. For both statistics, we established the square-root law, showing that the maximum expected number of covert jobs satisfies $E_{b,U}^{*}(n)=\Theta(\sqrt{n})$, for $U\in\{X,Y\}$. Moreover, inserting $\omega(\sqrt{n})$ jobs on average enables Willie to reliably detect the covert traffic. We further characterized the corresponding covert queuing scaling constants, $K_X$ and $K_Y$, and argued that $K_X\leq K_Y$, indicating that the number of jobs served is at least as informative as the busy-period length for detecting covert traffic.

Several directions remain open for future investigation. One direction is to analytically characterize the asymptotic behavior of $K_X$ and $K_Y$ as $\lambda_w\to 0$ and $\lambda_w\to\mu^{-}$, along with the corresponding asymptotics of $J(\lambda_w,\mu)$ and $C(\lambda_w,\mu)$. Other directions include extending the analysis to more general job-insertion mechanisms beyond Poisson arrivals and to multi-server queueing systems.
\vspace{2in}

\appendices

\section{Proof for \cref{prop:ebn_lambda_b}}
\label{app:D}

Under $H_1$, the total arrival rate is
$\lambda_w+\lambda_b(n)$. The expected number of jobs served during a
session is
\begin{align}
    \mathbb{E}[X]
    =
    \frac{\mu}
    {\mu-\bigl(\lambda_w+\lambda_b(n)\bigr)}.
\end{align}
Since each arrival is a Bob job with probability
$\lambda_b(n)/(\lambda_w+\lambda_b(n))$, the expected number of Bob's
jobs served during one session is
\begin{align}
    \frac{\lambda_b(n)}
    {\lambda_w+\lambda_b(n)}
    \mathbb{E}[X].
\end{align}
Therefore, over $n$ sessions,
\begin{align}
    E_{b,U}(n)
    &=
    n
    \frac{\lambda_b(n)}
    {\lambda_w+\lambda_b(n)}
    \mathbb{E}[X]
    \notag\\
    &=
    n\lambda_b(n)
    \frac{\mu}
    {\bigl(\lambda_w+\lambda_b(n)\bigr)
    \bigl(\mu-\lambda_w-\lambda_b(n)\bigr)}.
    \label{eq:expected_jobs_stat1}
\end{align}

By the uniform stability assumption,
\begin{align}
    \eta
    \leq
    \mu-\lambda_w-\lambda_b(n)
    <
    \mu-\lambda_w,
    \label{Eq:t1}
\end{align}
whereas $\lambda_b(n)\geq0$ gives
\begin{align}
    \lambda_w
    \leq
    \lambda_w+\lambda_b(n)
    \leq
    \mu-\eta.
    \label{Eq:t2}
\end{align}
Since $\lambda_w>0$, \cref{Eq:t1,Eq:t2} imply that
\begin{align}
    \frac{\mu}
    {\bigl(\lambda_w+\lambda_b(n)\bigr)
    \bigl(\mu-\lambda_w-\lambda_b(n)\bigr)}
    =
    \Theta(1).
    \label{eq:lambda_b_rate}
\end{align}
Substituting \eqref{eq:lambda_b_rate} into
\eqref{eq:expected_jobs_stat1} yields
\begin{align}
    E_{b,U}(n)
    =
    n\Theta\bigl(\lambda_b(n)\bigr),
\end{align}
which proves \eqref{eq:Eb_equiv}. Rearranging gives
\begin{align}
    \lambda_b(n)
    =
    \frac{1}{n}\Theta\bigl(E_{b,U}(n)\bigr),
\end{align}
and completes the proof of Proposition~\ref{prop:ebn_lambda_b}.
\hfill $\blacksquare$

\section{Proof of \cref{Proposition:stat1_achievability}}
\label{app:A}

From Proposition~\ref{prop:ebn_lambda_b}, the assumption $E_{b,X}(n)=\omega(\sqrt{n})$
implies that
\begin{align}
    \lambda_b(n)=\omega\left(\frac{1}{\sqrt{n}}\right),
    \label{eq:lambda_b_omega}
\end{align}
or, equivalently,
\begin{align}
    \sqrt{n}\lambda_b(n)\to\infty.
    \label{eq:sqrt_n_lambda}
\end{align}

Willie computes the test statistic
\begin{align}
    \overline{S}
    \triangleq
    \frac{1}{n}\sum_{i=1}^n X_i.
\end{align}
and compares against a sequence $t_n$ (defined later).
Under $H_0$, the mean and variance of $\overline{S}$ are given by
\begin{align}
\mathbb{E}_0[\overline{S}] & = \frac{\mu}{\mu - \lambda_w} \triangleq  m_0, \\
\mathrm{Var}_0[\overline{S}] & = \frac{1}{n}  \frac{\mu \lambda_w (\mu + \lambda_w)}{(\mu - \lambda_w)^3} \triangleq \frac{v_0}{n}
\end{align}
which shows that $v_0>0$ is a constant independent of $n$. Similarly, under
$H_1$,
\begin{align}
\mathbb{E}_1[\overline{S}] & = \frac{\mu}{\mu - \lambda_w - \lambda_b} \triangleq m_1, \\
\mathrm{Var}_1[\overline{S}] &  = \frac{1}{n}  \frac{\mu (\lambda_w + \lambda_b) (\mu + \lambda_w + \lambda_b)}{(\mu - \lambda_w - \lambda_b)^3} \triangleq \frac{v_1(n)}{n}.
\end{align}
We assume that the queue is uniformly stable, i.e., there exists an
$\eta>0$ and $n_0$ such that
\begin{align}
    \mu-\lambda_w-\lambda_b(n)\geq\eta, \quad n \geq n_0.
\end{align}
 Consequently,
$v_1(n)=O(1)$.

The separation between the means under the two hypotheses is
\begin{align}
    m_1-m_0
    &=
    \frac{\mu}
    {\mu-\lambda_w-\lambda_b(n)}
    -
    \frac{\mu}{\mu-\lambda_w}
    \notag\\
    &=
    \frac{\mu\lambda_b(n)}
    {(\mu-\lambda_w-\lambda_b(n))(\mu-\lambda_w)}
    \notag\\
    &=
    \Theta\bigl(\lambda_b(n)\bigr),
    \label{eq:mean_separation}
\end{align}
where the last equality follows from uniform stability of the queue.

Now choose a sequence $t_n>0$ such that
\begin{align}
    t_n = \omega\left(\frac{1}{\sqrt{n}}\right)
    \qquad\text{and}\qquad
    t_n
    =
    o\bigl(\lambda_b(n)\bigr).
    \label{eq:threshold_conditions}
\end{align}
Such a sequence exists because of
\eqref{eq:lambda_b_omega}. In other words, there exists a 
\begin{align}
   \delta_n = o(1) 
\end{align}
 such that 
\begin{align}
    t_n
    = \delta_n \lambda_b(n)
    \label{eq:threshold_choice}
\end{align}
Willie decides $H_1$ if
\begin{align}
    |\overline{S}-m_0|\geq t_n
\end{align}
and decides $H_0$ otherwise.

We first consider the probability of false alarm. By Chebyshev's
inequality~\cite[Sec 3.7]{Ibe2014},
\begin{align}
    \alpha_n
    &=
    \mathbb{P}_0
    \left(
        |\overline{S}-m_0|\geq t_n
    \right)
    \notag\\
    &\leq
    \frac{\operatorname{Var}_0[\overline{S}]}{t_n^2}
    \notag\\
    &=
    \frac{v_0}{nt_n^2}.
\end{align}
since $    t_n = \omega\left(\frac{1}{\sqrt{n}}\right)$, we have
\begin{align}
    nt_n^2\to\infty,
\end{align}
and hence
\begin{align}
    \alpha_n\to0.
    \label{eq:alpha_to_zero}
\end{align}

We next consider the probability of missed detection, which is given by
\begin{align}
    \beta_n
    &=
    \mathbb{P}_1
    \left(
        |\overline{S}-m_0|<t_n
    \right).
    \label{beta_bound1}
\end{align}
To bound this probability under $\mathbb{P}_1$, we express the
missed-detection event in terms of the deviation of $\overline{S}$
from its mean $m_1$.
By the triangle inequality~\cite[Theorem~4.3]{RudinRCA}
\[
m_1-m_0
=
\left|(\overline{S}-m_0)-(\overline{S}-m_1)\right|
\leq
|\overline{S}-m_0|+|\overline{S}-m_1|,
\]
which gives
\[
|\overline{S}-m_1|
\geq
m_1-m_0-|\overline{S}-m_0|.
\]
Thus,
\[
|\overline{S}-m_0|<t_n
\quad\Longrightarrow\quad
|\overline{S}-m_1|>m_1-m_0-t_n,
\]
and hence
\[
\left\{
|\overline{S}-m_0|<t_n
\right\}
\subseteq
\left\{
|\overline{S}-m_1|\geq m_1-m_0-t_n
\right\}.
\]
\color{black}
Therefore,
\begin{align}
    \beta_n
    &\leq
    \mathbb{P}_1
    \left(
        |\overline{S}-m_1|
        \geq m_1-m_0-t_n
    \right)
    \notag\\
    &\leq
    \frac{\operatorname{Var}_1[\overline{S}]}
    {(m_1-m_0-t_n)^2}
    \notag\\
    &=
    \frac{v_1(n)}
    {n(m_1-m_0-t_n)^2},
    \label{eq:beta_bound}
\end{align}
where the second inequality follows from Chebyshev's inequality under
$\mathbb{P}_1$.

Eq. \eqref{eq:mean_separation} and the fact that
$t_n=o(\lambda_b(n))$ imply that
\begin{align}
    m_1-m_0-t_n
    =
    \Theta\bigl(\lambda_b(n)\bigr).
\end{align}
Moreover, $v_1(n)=O(1)$ by uniform stability. Hence,
\begin{align}
    \beta_n
    &=
    O\left(
        \frac{1}{n\lambda_b^2(n)}
    \right).
\end{align}
From \eqref{eq:sqrt_n_lambda}, we have
\begin{align}
    n\lambda_b^2(n)
    =    \left(\sqrt{n}\lambda_b(n)\right)^2
    \to\infty,
\end{align}
and therefore
\begin{align}
    \beta_n\to 0
    \label{eq:beta_to_zero}
\end{align}
as $n \to \infty$.

Combining \eqref{eq:alpha_to_zero} and
\eqref{eq:beta_to_zero}, we obtain
\begin{align}
    \alpha_n+\beta_n\to0,
\end{align}
which completes the proof of Proposition~\ref{Proposition:stat1_achievability}.
\hfill $\blacksquare$

\section{Taylor series approximation for KL-divergence of Statistic-1}
\label{app:E}

We find the second-order Taylor series approximation for $D_{\mathrm{KL}}(f_0\|f_1)$, where the density $f_j(\cdot), j=0,1$ is given by

\begin{equation}
\label{eq:number_stat1_forkl}
f_j(x)
=
\frac{1}{x}
\binom{2x-2}{x-1}
\frac{\rho_j^{x-1}}{(1+\rho_j)^{2x-1}},
\qquad x=1,2,\ldots,
\end{equation}
where
\[
\rho_0=\frac{\lambda_w}{\mu},
\qquad
\rho_1=\frac{\lambda_w+\lambda_b}{\mu},
\]
under $H_0$ and $H_1$, respectively. We obtain
\begin{align}
\frac{f_0(x)}{f_1(x)}
&=
\left(\frac{\rho_0}{\rho_1}\right)^{x-1}
\left(\frac{1+\rho_1}{1+\rho_0}\right)^{2x-1}
\nonumber\\
&=c_1^{x-1}c_2^{2x-1},
\label{eq:number_stat1_likelihood_ratio}
\end{align}
where
\[
c_1\triangleq \frac{\lambda_w}{\lambda_w+\lambda_b},
\qquad
c_2\triangleq\frac{\mu+\lambda_w+\lambda_b}{\mu+\lambda_w}.
\]
Therefore, the KL divergence is
\begin{align}
D_{\mathrm{KL}}(f_0\|f_1)
&=
\sum_{x=1}^{\infty}
f_0(x)
\log\frac{f_0(x)}{f_1(x)}
\nonumber\\
&=
\log(c_1)\sum_{x=1}^{\infty}f_0(x)(x-1)
\nonumber\\
&+
\log(c_2)\sum_{x=1}^{\infty}f_0(x)(2x-1).
\label{eq:number_stat1_kl}
\end{align}
From the distribution in \eqref{eq:expectation_stat1},
\[
\mathbb{E}_0[X]
=
\frac{1}{1-\rho_0}.
\]
Thus,
\[
\sum_{x=1}^{\infty}f_0(x)(x-1)
=
\frac{\rho_0}{1-\rho_0}
=
\frac{\lambda_w}{\mu-\lambda_w},
\]
and
\[
\sum_{x=1}^{\infty}f_0(x)(2x-1)
=
\frac{1+\rho_0}{1-\rho_0}
=
\frac{\mu+\lambda_w}{\mu-\lambda_w}.
\]
Consequently,
\begin{equation}
\label{eq:number_stat1_kl_exact}
D_{\mathrm{KL}}(f_0\|f_1)
=
\log(c_1)\frac{\lambda_w}{\mu-\lambda_w}
+
\log(c_2)\frac{\mu+\lambda_w}{\mu-\lambda_w}.
\end{equation}
To obtain the second-order Taylor approximation, we expand the
logarithmic terms around $\lambda_b=0$. Since
\[
\log(1+x)=x-\frac{x^2}{2}+O(x^3),
\]
we have
\begin{equation}
\log(c_1)
=
-\log\left(1+\frac{\lambda_b}{\lambda_w}\right)
=
-\frac{\lambda_b}{\lambda_w}
+\frac{\lambda_b^2}{2\lambda_w^2}
+O(\lambda_b^3),
\label{eq:number_stat1_logc1}
\end{equation}
\begin{align}
\log(c_2)
&=
\log\left(1+\frac{\lambda_b}{\mu+\lambda_w}\right)
\nonumber\\
&=
\frac{\lambda_b}{\mu+\lambda_w}
-\frac{\lambda_b^2}{2(\mu+\lambda_w)^2}
+O(\lambda_b^3).
\label{eq:number_stat1_logc2}
\end{align}

Substituting \eqref{eq:number_stat1_logc1} and
\eqref{eq:number_stat1_logc2} into \eqref{eq:number_stat1_kl_exact},
we obtain
\begin{align}
D_{\mathrm{KL}}(f_0\|f_1)
&=
\frac{\lambda_b^2}{2(\mu-\lambda_w)}
\left(
\frac{1}{\lambda_w}
-\frac{1}{\mu+\lambda_w}
\right)
+O(\lambda_b^3)
\nonumber\\
&=
\frac{\mu}
{2\lambda_w(\mu^2-\lambda_w^2)}
\lambda_b^2
+O(\lambda_b^3).
\label{eq:number_stat1_kl_taylor}
\end{align}
This completes the proof of Lemma~\ref{lem:stat1_kl}.
\hfill $\blacksquare$

\section{Proof of \cref{Proposition:stat2}}
\label{app:B}

\subsection{Proof of Achievability for statistic $U=Y$}

Recall that $Y_1,Y_2,\ldots,Y_n$ is the sequence of random
variables corresponding to the lengths of $n$ successive busy
periods and $g_i(y)$ denotes the density of $Y_j$ under
$H_i$ for $i=0,1$ and $j=1,2,\ldots,n$.

From \cref{lem:Tayor_D_KL_Stat2}, we
have
\begin{equation}
D_{\mathrm{KL}}(g_0\|g_1)
=
C(\lambda_w, \mu)\lambda_b^2+o(\lambda_b^2),
\label{eq:kl_stat2_ach}
\end{equation}
For $n$ independent busy periods, the KL divergence is 
\begin{equation}
D_{\mathrm{KL}}(g_0^n\|g_1^n)
=
nD_{\mathrm{KL}}(g_0\|g_1).
\end{equation}
For covertness, from \cref{covert_defibition}
\begin{equation}
nD_{\mathrm{KL}}(g_0,g_1)\leq\delta.
\end{equation}
From \eqref{eq:kl_stat2_ach},
\begin{equation}
n\left(C(\lambda_w, \mu)\lambda_b^2+o(\lambda_b^2)\right)
\leq\delta.
\label{eq:intermediate_step}
\end{equation}
Simplifying \eqref{eq:intermediate_step}, asymptotically, for covertness $\lambda_b$ should satisfy
\begin{equation}
\lambda_b
\leq
\sqrt{
\frac{\delta}
{nC(\lambda_w, \mu)}
}.
\end{equation}
Thus, an insertion rate satisfying
\begin{equation}
\lambda_b=O\left(\frac{1}{\sqrt{n}}\right)
\end{equation}
is achievable.
To establish the achievability result, choose the insertion rate as
\begin{equation}
\lambda_b
=
\frac{\eta\sqrt{\delta}}
{\sqrt{nC(\lambda_w,\mu)}},
\qquad 0<\eta<1.
\label{eq:lambda_b_choice}
\end{equation}
By \cref{prop:ebn_lambda_b},
\begin{equation}
    E_{b,Y}(n)=n\Theta(\lambda_b).
\end{equation}
Hence,
\begin{equation}
    E_{b,Y}(n)=\Theta(\sqrt{n}),
\end{equation}
which proves the achievability part of \cref{Proposition:stat2}
\hfill $\blacksquare$

\subsection{Proof of Converse}

From Proposition~\ref{prop:ebn_lambda_b}, the assumption $E_{b,Y}(n)=\omega(\sqrt{n})$
implies that 
\begin{align}
    \lambda_b(n)=\omega\left(\frac{1}{\sqrt{n}}\right),
    \label{eq:lambda_b_omega_stat2}
\end{align}
or, equivalently,
\begin{align}
    \sqrt{n}\lambda_b(n)\to\infty.
    \label{eq:sqrt_n_lambda_stat2}
\end{align}

Willie computes the test statistic
\begin{align}
    \overline{S}
    \triangleq
    \frac{1}{n}\sum_{i=1}^n Y_i.
\end{align}
and compares against a sequence $t_n$ (defined later).
Under $H_0$, the mean and variance of $\overline{S}$ are given by
\begin{align}
\label{stat_2_mean}
\mathbb{E}_0[\overline{S}] & = \frac{1}{\mu-\lambda_w} \triangleq  m_0, \\
\label{stat_2_variance}
\mathrm{Var}_0[\overline{S}] & = \frac{1}{n}
\frac{\mu+\lambda_w}
{(\mu-\lambda_w)^3} \triangleq \frac{v_0}{n}
\end{align}
which shows that $v_0>0$ is a constant independent of $n$. Similarly, under
$H_1$,
\begin{align}
\mathbb{E}_1[\overline{S}] & = \frac{1}
{\mu-\lambda_w-\lambda_b} \triangleq m_1, \\
\mathrm{Var}_1[\overline{S}] &  = \frac{1}{n}
\frac{\mu+\lambda_w+\lambda_b}
{(\mu-\lambda_w-\lambda_b)^3} \triangleq \frac{v_1(n)}{n}.
\end{align}
We assume that the queue is uniformly stable, i.e., there exists an
$\eta>0$ and $n_0$ such that
\begin{align}
    \mu-\lambda_w-\lambda_b(n)\geq\eta, \quad n \geq n_0.
\end{align}
 Consequently,
$v_1(n)=O(1)$.

The separation between the means under the two hypotheses is
\begin{align}
    m_1-m_0
    &=
    \frac{1}
    {\mu-\lambda_w-\lambda_b(n)}
    -
    \frac{1}{\mu-\lambda_w}
    \notag\\
    &=
    \frac{\lambda_b(n)}
    {(\mu-\lambda_w-\lambda_b(n))(\mu-\lambda_w)}
    \notag\\
    &=
    \Theta\bigl(\lambda_b(n)\bigr),
    \label{eq:mean_separation_stat2}
\end{align}
where the last equality follows from uniform stability of the queue.

Now choose a sequence $t_n>0$ such that
\begin{align}
    t_n = \omega\left(\frac{1}{\sqrt{n}}\right)
    \qquad\text{and}\qquad
    t_n
    =
    o\bigl(\lambda_b(n)\bigr).
    \label{eq:threshold_conditions_stat2}
\end{align}
Such a sequence exists because of
\eqref{eq:lambda_b_omega_stat2}. In other words, there exists a 
\begin{align}
   \delta_n = o(1) 
\end{align}
 such that 
\begin{align}
    t_n
    = \delta_n \lambda_b(n)
    \label{eq:threshold_choice_stat2}
\end{align}
Willie decides $H_1$ if
\begin{align}
    |\overline{S}-m_0|\geq t_n
\end{align}
and decides $H_0$ otherwise.

We first consider the probability of false alarm. By Chebyshev's
inequality~\cite[Sec 3.7]{Ibe2014},
\begin{align}
    \alpha_n
    &=
    \mathbb{P}_0
    \left(
        |\overline{S}-m_0|\geq t_n
    \right)
    \notag\\
    &\leq
    \frac{\operatorname{Var}_0[\overline{S}]}{t_n^2}
    \notag\\
    &=
    \frac{v_0}{nt_n^2}.
\end{align}
Since $    t_n = \omega\left(\frac{1}{\sqrt{n}}\right)$, we have
\begin{align}
    nt_n^2\to\infty,
\end{align}
and hence
\begin{align}
    \alpha_n\to0.
    \label{eq:alpha_to_zero_stat2}
\end{align}

We next consider the probability of missed detection, which is given by
\begin{align}
    \beta_n
    &=
    \mathbb{P}_1
    \left(
        |\overline{S}-m_0|<t_n
    \right).
    \label{beta_bound1_stat2}
\end{align}
To bound this probability under $\mathbb{P}_1$, we express the
missed-detection event in terms of the deviation of $\overline{S}$
from its mean $m_1$.
By the triangle inequality~\cite[Theorem~4.3]{RudinRCA},
\[
m_1-m_0
=
\left|(\overline{S}-m_0)-(\overline{S}-m_1)\right|
\leq
|\overline{S}-m_0|+|\overline{S}-m_1|,
\]
which gives
\[
|\overline{S}-m_1|
\geq
m_1-m_0-|\overline{S}-m_0|.
\]
Thus,
\[
|\overline{S}-m_0|<t_n
\quad\Longrightarrow\quad
|\overline{S}-m_1|>m_1-m_0-t_n,
\]
and hence
\[
\left\{
|\overline{S}-m_0|<t_n
\right\}
\subseteq
\left\{
|\overline{S}-m_1|\geq m_1-m_0-t_n
\right\}.
\]
\color{black}
Therefore,
\begin{align}
    \beta_n
    &\leq
    \mathbb{P}_1
    \left(
        |\overline{S}-m_1|
        \geq m_1-m_0-t_n
    \right)
    \notag\\
    &\leq
    \frac{\operatorname{Var}_1[\overline{S}]}
    {(m_1-m_0-t_n)^2}
    \notag\\
    &=
    \frac{v_1(n)}
    {n(m_1-m_0-t_n)^2},
    \label{eq:beta_bound_stat2}
\end{align}
where the second inequality follows from Chebyshev's inequality under
$\mathbb{P}_1$.
Eq. \eqref{eq:mean_separation_stat2} and the fact that
$t_n=o(\lambda_b(n))$ imply that
\begin{align}
    m_1-m_0-t_n
    =
    \Theta\bigl(\lambda_b(n)\bigr).
\end{align}
Moreover, $v_1(n)=O(1)$ by uniform stability. Hence,
\begin{align}
    \beta_n
    &=
    O\left(
        \frac{1}{n\lambda_b^2(n)}
    \right).
\end{align}

From \eqref{eq:sqrt_n_lambda_stat2}, we have
\begin{align}
    n\lambda_b^2(n)
    =
    \left(\sqrt{n}\lambda_b(n)\right)^2
    \to\infty,
\end{align}
and therefore
\begin{align}
    \beta_n\to 0
    \label{eq:beta_to_zero_stat2}
\end{align}
as $n \to \infty$.

Combining \eqref{eq:alpha_to_zero_stat2} and
\eqref{eq:beta_to_zero_stat2}, we obtain
\begin{align}
    \alpha_n+\beta_n\to0,
\end{align}
which completes the proof of Proposition~\ref{Proposition:stat2}.
\hfill $\blacksquare$

\section{Proof oF Lemma~\ref{lem:Tayor_D_KL_Stat2}}
\label{app:C}

Let
\begin{equation}
g_i(y)
=
\frac{e^{-(\xi_i+\mu)y}}{y}
\sqrt{\frac{\mu}{\xi_i}}
I_1\!\left(2y\sqrt{\mu\xi_i}\right),
\qquad y>0,
\end{equation}
where
\[
\xi_0=\lambda_w,
\qquad
\xi_1=\lambda_w+\lambda_b.
\]
Throughout the proof, we assume $0<\lambda_w<\mu$, and fix $0 < \eta < \mu -\lambda_w$. We restrict $0 \leq \lambda_b \leq \eta$ which is sufficient for the expansion around $\lambda_b =0$. For convenience, we define
\begin{equation}
u_0=2y\sqrt{\mu\lambda_w},
\qquad
u_1=2y\sqrt{\mu(\lambda_w+\lambda_b)}.
\label{u0u1}
\end{equation}

We regard the KL divergence as a function of $\lambda_b$, and hence we define
\begin{align}
D(\lambda_b)
&\triangleq
D_{\mathrm{KL}}(g_0\|g_1) \notag\\
&=
\int_0^\infty
g_0(y)
\log
\frac{g_0(y)}{g(y,\lambda_b)}
\,dy,
\label{eq:f_lambda}
\end{align}
where $g(y,\lambda_b)=g_1(y)$. Since $g(y,0)=g_0(y)$, we have
\begin{equation}
D(0)=0.
\end{equation}
We next evaluate the first two derivatives of
$D(\lambda_b)$ at $\lambda_b=0$. To this end, we define
\begin{align}
h(y,\lambda_b)
\triangleq
g_0(y)
\log\frac{g_0(y)}{g(y,\lambda_b)}.
\end{align}
Then, we have
\begin{align}
\frac{\partial h(y,\lambda_b)}{\partial\lambda_b}
=
-g_0(y)
\frac{\partial}{\partial\lambda_b}
\log g(y,\lambda_b).
\end{align}
Using
\begin{align}
\log g(y,\lambda_b)
={}&
-(\mu+\lambda_w+\lambda_b)y
-\log y
\nonumber\\
&+\frac12\log\mu -\frac12\log(\lambda_w+\lambda_b)
+\log I_1(u_1),
\end{align}
we obtain
\begin{align}
\frac{\partial h(y,\lambda_b)}{\partial\lambda_b}
=
g_0(y)
\left[
y+\frac{1}{2(\lambda_w+\lambda_b)}
-
y\sqrt{\frac{\mu}{\lambda_w+\lambda_b}}
\frac{I_1'(u_1)}{I_1(u_1)}
\right].
\label{diffh}
\end{align}

We have the following Bessel identity from~\cite[Eq.~(9.6.26)]{AbramowitzStegun}
\begin{equation}
\label{eq:Bessel_identity}
\frac{I_1'(x)}{I_1(x)}
=
\frac{I_0(x)}{I_1(x)}-\frac1x
\end{equation}
and the bound~\cite[Lemma~2]{Henderson2000}
\begin{equation}
\label{eq:Bessel_bound}
\frac{I_0(x)}{I_1(x)}
<
1+\frac{2}{x},
\qquad x>0.
\end{equation}
From \eqref{eq:Bessel_identity} and \eqref{eq:Bessel_bound},
we obtain, for $0\leq\lambda_b\leq\eta$,
\begin{align}
\left|
\frac{\partial h(y,\lambda_b)}
{\partial\lambda_b}
\right|
\leq
g_0(y)
\left[
\left(1+\sqrt{\frac{\mu}{\lambda_w}}\right)y
+\frac{3}{2\lambda_w}
\right].
\label{eq:dominating_h}
\end{align}
Since
\begin{equation}
\int_0^\infty g_0(y)\,dy=1,
\qquad
\int_0^\infty yg_0(y)\,dy
=
\frac{1}{\mu-\lambda_w}.
\end{equation}
the right-hand side (RHS) of 
\eqref{eq:dominating_h} is integrable. Hence, by the dominated convergence theorem~\cite[Theorem~1.34]{RudinRCA}, differentiation under the integral sign is justified, and we obtain
\begin{align}
D'(\lambda_b)
={}&
\frac{1}{\mu-\lambda_w}
+
\frac{1}{\lambda_w+\lambda_b} 
\nonumber\\
&-
\frac{\sqrt{\mu}}
{\sqrt{\lambda_w+\lambda_b}}
\int_0^\infty
yg_0(y)
\frac{I_0(u_1)}{I_1(u_1)}
\,dy.
\label{eq:fprime_final}
\end{align}

Substituting the expression for $g_0(y)$, define
\begin{equation}
F(y,\lambda_b)
\triangleq
e^{-(\mu+\lambda_w)y}
I_1(u_0)
\frac{I_0(u_1)}{I_1(u_1)}.
\label{eq:F_definition}
\end{equation}
The dominated convergence theorem~\cite[Theorem~1.34]{RudinRCA}
allows us to interchange the limit $\lambda_b\to0$ and the integral.
Since $u_1\to u_0$ as $\lambda_b\to0$, we have
\begin{equation}
\lim_{\lambda_b\to0}F(y,\lambda_b)
=
e^{-(\mu+\lambda_w)y}I_0(u_0).
\label{eq:F_limit}
\end{equation}

Therefore, evaluating \eqref{eq:fprime_final} at $\lambda_b=0$ gives
\begin{align}
D'(0)
={}&
\frac{1}{\mu-\lambda_w}
+
\frac{1}{\lambda_w}
\nonumber\\
&-
\frac{\mu}{\lambda_w}
\int_0^\infty
\,
e^{-(\mu+\lambda_w)y}
I_0(u_0)\,dy.
\label{eq:Dprime_zero}
\end{align}
Using the Bessel integral identity in~\cite[Eq.~(6.611.1)]{GradshteynRyzhik}
\begin{equation}
\int_0^\infty e^{-\alpha y} I_0(\beta y)\,dy
=
\frac{1}{\sqrt{\alpha^2-\beta^2}},
\qquad \alpha>|\beta|.
\label{eq:bessel_integral}
\end{equation} for 
\[
\alpha=\mu+\lambda_w,
\qquad
\beta=2\sqrt{\mu\lambda_w},
\]
we obtain
\begin{align}
\int_0^\infty
e^{-(\mu+\lambda_w)y}
I_0\left(2y\sqrt{\mu\lambda_w}\right)
\,dy
&=
\frac{1}
{\sqrt{(\mu+\lambda_w)^2-4\mu\lambda_w}}
\nonumber\\
&=
\frac{1}{\mu-\lambda_w},
\end{align}
where we used $0<\lambda_w<\mu$. Consequently,
\begin{align}
D'(0)
&=
\frac{1}{\mu-\lambda_w}
+
\frac{1}{\lambda_w}
-
\frac{\mu}
{\lambda_w(\mu-\lambda_w)}
\nonumber\\
&=0.
\label{eq:Dprime_zero_final}
\end{align}
\color{black}
For
\begin{align}
D'(\lambda_b)
=
\frac{1}{\mu-\lambda_w}
+
\frac{1}{\lambda_w+\lambda_b}
-
A(\lambda_b)B(\lambda_b),
\label{eq:Dprime_AB}
\end{align}
where
\begin{equation}
A(\lambda_b)
\triangleq
\frac{\mu}
{\sqrt{\lambda_w(\lambda_w+\lambda_b)}}
\end{equation}
and
\begin{equation}
B(\lambda_b)
\triangleq
\int_0^\infty F(y,\lambda_b)\,dy.
\end{equation}

We next evaluate the second derivative. Differentiating
\eqref{eq:Dprime_AB} gives
\begin{align}
D''(\lambda_b)
=&
-\frac{1}{(\lambda_w+\lambda_b)^2}
+
\frac{\mu}
{2\sqrt{\lambda_w}(\lambda_w+\lambda_b)^{3/2}}
B(\lambda_b) 
\nonumber\\
&-
\frac{\mu}
{\sqrt{\lambda_w(\lambda_w+\lambda_b)}}
B'(\lambda_b).
\label{eq:Dsecond_AB}
\end{align}

To evaluate \(B'(\lambda_b)\), observe that
\begin{align}
\frac{\partial}{\partial\lambda_b}
\left(
\frac{I_0(u_1)}{I_1(u_1)}
\right)
=
\frac{\partial u_1}{\partial\lambda_b}
\left[
1-
\left(\frac{I_0(u_1)}{I_1(u_1)}\right)^2
+
\frac{I_0(u_1)}{u_1I_1(u_1)}
\right],
\end{align}
where
\begin{equation}
\frac{\partial u_1}{\partial\lambda_b}
=
y\sqrt{\frac{\mu}{\lambda_w+\lambda_b}}.
\end{equation}
Using the Bessel-function bounds in \eqref{eq:Bessel_identity} and \eqref{eq:Bessel_bound}, we obtain
\begin{align}
\left|
\frac{\partial F(y,\lambda_b)}
{\partial\lambda_b}
\right|
\leq
\mu
e^{-q_1y}
\left[
2y^2+
\frac{5y}{2\sqrt{\mu\lambda_w}}
+
\frac{3}{2\mu\lambda_w}
\right],
\label{eq:F_derivative_bound}
\end{align}
where
\begin{equation}
q_1
=
\mu+\lambda_w-2\sqrt{\mu\lambda_w}
=
(\sqrt{\mu}-\sqrt{\lambda_w})^2>0.
\end{equation}
The RHS of \eqref{eq:F_derivative_bound} is integrable and
independent of $\lambda_b$. Therefore, differentiation under the
integral sign is justified. After expanding $B'(\lambda_b)$, we obtain
\begin{align}
B'(\lambda_b)
=&
\sqrt{\frac{\lambda_w}{\lambda_w+\lambda_b}}
\frac{2\mu}{(\mu-\lambda_w)^3}
-
\sqrt{\frac{\mu}{\lambda_w+\lambda_b}}
Q(\lambda_b)
\nonumber\\
&+
\frac{1}{2(\lambda_w+\lambda_b)}
B(\lambda_b),
\label{eq:Bprime_terms}
\end{align}
where
\begin{equation}
Q(\lambda_b)
\triangleq
\int_0^\infty
y e^{-(\mu+\lambda_w)y}
I_1(u_0)
\left(
\frac{I_0(u_1)}{I_1(u_1)}
\right)^2dy.
\label{eq:Q_definition}
\end{equation}
Here we used
\begin{equation}
\mathbb{E}_{g_0}[Y^2]
=
\frac{2\mu}{(\mu-\lambda_w)^3}.
\end{equation}

Substituting \eqref{eq:Bprime_terms} into
\eqref{eq:Dsecond_AB}, the terms containing $B(\lambda_b)$ cancel,
yielding
\begin{align}
D''(\lambda_b)
=&
-\frac{1}{(\lambda_w+\lambda_b)^2}
-
\frac{2\mu^2}
{(\lambda_w+\lambda_b)(\mu-\lambda_w)^3} 
\nonumber\\
&+
\frac{\mu^{3/2}}
{\sqrt{\lambda_w}(\lambda_w+\lambda_b)}
Q(\lambda_b).
\label{eq:Dsecond_simplified}
\end{align}

It remains to evaluate \(Q(\lambda_b)\) as
\(\lambda_b\to0\). Since \(u_1\geq u_0\), the Bessel function bounds give
\begin{equation}
\begin{aligned}
y e^{-(\mu+\lambda_w)y} I_1(u_0)
\left(\frac{I_0(u_1)}{I_1(u_1)}\right)^2
&\leq
\sqrt{\mu\lambda_w}\,e^{-q_1y}
\left(
y^2+\frac{2y}{\sqrt{\mu\lambda_w}}
\right.
\\[-1mm]
&\qquad\left.
+\frac{1}{\mu\lambda_w}
\right).
\end{aligned}
\label{eq:Bessel_upper}
\end{equation}
The RHS of \eqref{eq:Bessel_upper} is integrable because $q_1>0$.
Hence, by the dominated convergence theorem, it follows that
\begin{align}
\lim_{\lambda_b\to0}Q(\lambda_b)
&=
\int_0^\infty
y e^{-(\mu+\lambda_w)y}
\frac{I_0^2(u_0)}{I_1(u_0)}
\,dy \notag\\
&\triangleq
J(\lambda_w,\mu).
\end{align}
Therefore,
\begin{align}
D''(0)
=
-\frac{1}{\lambda_w^2}
-\frac{2\mu^2}
{\lambda_w(\mu-\lambda_w)^3}
+
\frac{\mu^{3/2}}{\lambda_w^{3/2}}
J(\lambda_w,\mu).
\end{align}

By defining
\begin{align}
C(\lambda_w,\mu)
\triangleq
\frac12
\left[
-\frac{1}{\lambda_w^2}
-\frac{2\mu^2}
{\lambda_w(\mu-\lambda_w)^3}
+
\frac{\mu^{3/2}}{\lambda_w^{3/2}}
J(\lambda_w,\mu)
\right],
\end{align}
Since $D(\lambda_b)=D_{\mathrm{KL}}(g_0\|g_1)\geq 0$ for all
$\lambda_b\geq 0$, and $D(0)=D'(0)=0$, the quadratic coefficient
satisfies $C(\lambda_w,\mu)\geq 0$.
Thus, the second-order Taylor expansion around $\lambda_b=0$ yields
\begin{equation}
D_{\mathrm{KL}}(g_0\|g_1)
=C(\lambda_w,\mu)\lambda_b^2+o(\lambda_b^2).
\end{equation}
which completes the proof of Lemma~\ref{lem:Tayor_D_KL_Stat2}.

To show that $C(\lambda_w,\mu)>0$, define the score function
\begin{align}
s(y)
&\triangleq
\left.
\frac{\partial}{\partial \xi}
\log g_\xi(y)
\right|_{\xi=\lambda_w}.
\label{eq:score_function}
\end{align}
By the differentiability established above,
\begin{align}
2C(\lambda_w,\mu)
&=
\left.
\frac{\partial^2}{\partial\lambda_b^2}
D_{\mathrm{KL}}(g_0\|g_1)
\right|_{\lambda_b=0}
=
\mathbb{E}_0[s(Y)^2].
\label{eq:C_fisher_information}
\end{align}

From \eqref{stat_2_mean} and \eqref{stat_2_variance}
\begin{align}
\mathbb{E}_\xi[Y]
&=
\frac{1}{\mu-\xi},
&
\operatorname{Var}_\xi(Y)
&=
\frac{\mu+\xi}{(\mu-\xi)^3},
\end{align}
differentiating $\mathbb{E}_\xi[Y]$ with respect to $\xi$ gives
\begin{align}
\mathbb{E}_0[Ys(Y)]
&=
\left.
\frac{d}{d\xi}\mathbb{E}_\xi[Y]
\right|_{\xi=\lambda_w}
=
\frac{1}{(\mu-\lambda_w)^2}.
\label{eq:score_moment}
\end{align}
Moreover, $\mathbb{E}_0[s(Y)]=0$, and hence
\begin{align}
\mathbb{E}_0
\left[
\bigl(Y-\mathbb{E}_0[Y]\bigr)s(Y)
\right]
&=
\frac{1}{(\mu-\lambda_w)^2}.
\label{eq:centered_score_moment}
\end{align}

By the Cauchy--Schwarz inequality~\cite[Theorem~4.2]{RudinRCA},
\begin{align}
\mathbb{E}_0[s(Y)^2]
&\geq
\frac{
\mathbb{E}_0
\left[
\bigl(Y-\mathbb{E}_0[Y]\bigr)s(Y)
\right]^2
}{
\operatorname{Var}_0(Y)
}
\nonumber\\
&=
\frac{1}{\mu^2-\lambda_w^2}.
\label{eq:fisher_lower_bound}
\end{align}
Combining \eqref{eq:C_fisher_information} and
\eqref{eq:fisher_lower_bound} gives
\begin{align}
C(\lambda_w,\mu)
&\geq
\frac{1}{2(\mu^2-\lambda_w^2)}
>0.
\label{eq:C_strictly_positive}
\end{align}

\hfill $\blacksquare$

\section{Proof of Theorem~\ref{Thm:KY}}
\label{app:proof_thm_KY}

First, define
\begin{align}
    c_Y \triangleq C(\lambda_w,\mu),
\end{align}
where $C(\lambda_w,\mu)>0$ is given in
\eqref{Eq:C_lambda_mu_remark}.
From Lemma~\ref{lem:Tayor_D_KL_Stat2}, as $\lambda_b\to 0$,
\begin{align}
D_{\mathrm{KL}}(g_0\|g_1)
=
c_Y\lambda_b^2+o(\lambda_b^2).
\label{eq:KY_KL_expansion}
\end{align}
Moreover, the expected number of additional jobs
inserted by Bob over $n$ sessions is
\begin{align}
E_{b,Y}(n)
&=
n\frac{\lambda_b(n)}
{\lambda_w+\lambda_b(n)}\mathbb{E}[X]
\nonumber\\
&=
\frac{n\mu\lambda_b(n)}
{(\lambda_w+\lambda_b(n))
(\mu-\lambda_w-\lambda_b(n))}.
\label{eq:KY_expected_jobs}
\end{align}

\noindent\textit{Achievability:}
Choose the insertion rate
\begin{align}
\lambda_b(n)
=
\sqrt{\frac{\delta}{nc_Y}}.
\label{eq:KY_ach_rate}
\end{align}
Since the session statistics are i.i.d.,
\eqref{eq:KY_KL_expansion} yields
\begin{align}
D_{\mathrm{KL}}(g_0^{n}\|g_1^{n})
&=
nD_{\mathrm{KL}}(g_0\|g_1)
\nonumber\\
&=
nc_Y\lambda_b(n)^2+o(1)
\nonumber\\
&=
\delta+o(1).
\label{eq:KY_ach_KL}
\end{align}
Thus, the asymptotic KL-divergence covertness
constraint is satisfied.

Furthermore, substituting
\eqref{eq:KY_ach_rate} into
\eqref{eq:KY_expected_jobs}, we obtain
\begin{align}
\lim_{n\to\infty}
\frac{E_{b,Y}(n)}{\sqrt{n\delta}}
&=
\frac{\mu}
{\lambda_w(\mu-\lambda_w)\sqrt{c_Y}}
\nonumber\\
&=
\frac{\mu}
{\lambda_w(\mu-\lambda_w)
\sqrt{C(\lambda_w,\mu)}}.
\label{eq:KY_ach_bound}
\end{align}
This establishes the achievability bound.

\medskip
\noindent\textit{Converse:}
Consider any sequence of insertion rates
$\{\lambda_b(n)\}_{n\geq 1}$ satisfying
the asymptotic KL-divergence covertness
constraint
\begin{align}
\limsup_{n\to\infty}
D_{\mathrm{KL}}
(g_0^{ n}\|g_1^{n})
\leq \delta.
\label{eq:KY_converse_constraint}
\end{align}
Since the KL divergence vanishes only when
$\lambda_b=0$ and is bounded away from zero
for insertion rates bounded away from zero,
\eqref{eq:KY_converse_constraint} implies
\begin{align}
\lambda_b(n)\to 0.
\end{align}

Consequently, from
\eqref{eq:KY_KL_expansion}, we obtain
\begin{align}
\limsup_{n\to\infty}
n\lambda_b(n)^2
\leq
\frac{\delta}{c_Y}.
\label{eq:KY_converse_rate}
\end{align}

Using \eqref{eq:KY_expected_jobs} and
\eqref{eq:KY_converse_rate}, it follows that
\begin{align}
\limsup_{n\to\infty}
\frac{E_{b,Y}(n)}{\sqrt{n\delta}}
&\leq
\frac{\mu}
{\lambda_w(\mu-\lambda_w)\sqrt{c_Y}}
\nonumber\\
&=
\frac{\mu}
{\lambda_w(\mu-\lambda_w)
\sqrt{C(\lambda_w,\mu)}}.
\label{eq:KY_converse_bound}
\end{align}

Combining the achievability and converse
bounds completes the proof of Theorem~\ref{Thm:KY}.
\hfill $\blacksquare$

\bibliographystyle{IEEEtran}
\bibliography{References}

\end{document}